\documentclass[10pt,twocolumn,conference,final]{IEEEtran}

\usepackage[small]{caption}
\usepackage{graphicx}
\usepackage{ifthen,color}
\usepackage{float,times}
\usepackage{amsmath,comment}
\usepackage{amssymb}
\usepackage{amscd}
\usepackage{amsfonts}
\usepackage{graphicx}
\usepackage{theorem}
\usepackage{epsfig,rotating}
\usepackage{colordvi}           
\usepackage{algorithm} 
\usepackage{algorithmic}
\usepackage{subfig}

\newtheorem{thm}{Theorem}
\newtheorem{cor}{Corollary}
\newtheorem{lem}{Lemma}
\newtheorem{defn}{Definition}

\newtheorem{rem}{Remark}

\newtheorem{prob}{Problem}

 \def\margine#1{\strut\vadjust{\kern-\strutdepth\specialstar{#1}}}
 \def\strutdepth{\dp\strutbox}
 \def\specialstar#1{\vtop to \strutdepth{\baselineskip\strutdepth\vss\llap{#1}\null}}
 
\newcommand{\Ng}[1]{{N_g(#1)}}

\newcommand{\Z}{{\mathbb Z}}
\newcommand{\ZnZ}[1]{\Z/#1\Z}
\IEEEoverridecommandlockouts
\begin{document}

\title{Distributed Wake-Up Scheduling for Energy Saving in Wireless Networks}
\author{Francesco De Pellegrini, Karina Gomez, Daniele Miorandi and Imrich Chlamtac\thanks{The authors are with CREATE-NET,  via alla Cascata 56/D, 38123 -- Povo, Trento, IT.}}
%define the running headings

\maketitle
%%%%%%%%%%%%%%%%%%%%%%%%%%%%%%%%%%%%%%%%%%%%%%%%%%%%%%%%%%%%%%%%%%%%%%%%%%%%%%%%%%%%%%%%%%%%%%%%%%%%%%%%%%%%%%%%%%%%%%%%
% ABSTRACT
%%%%%%%%%%%%%%%%%%%%%%%%%%%%%%%%%%%%%%%%%%%%%%%%%%%%%%%%%%%%%%%%%%%%%%%%%%%%%%%%%%%%%%%%%%%%%%%%%%%%%%%%%%%%%%%%%%%%%%%%

%define the abstract

\begin{abstract}
A customary solution to reduce the energy consumption of wireless communication
devices is to periodically put the radio into low-power sleep mode. A
relevant problem is to %envisage appropriate methods for scheduling the
schedule the wake-up of nodes in such a way as to ensure proper coordination among
devices, respecting delay constraints while still saving energy. In
this paper, we introduce a simple algebraic characterisation of the 
problem of periodic wake-up scheduling under both energy consumption 
and delay constraints. We demonstrate that the general problem of 
wake-up times coordination is equivalent to integer factorization 
and discuss the implications on the design of efficient scheduling 
algorithms. We then propose simple polynomial time heuristic algorithms that 
can be implemented in a distributed fashion and present a message 
complexity of the order of the number of links in the network.
 
Numerical results are provided in order to assess the performance of
the proposed techniques when applied to wireless sensor networks.

\end{abstract}

%define the keywords


%%%%%%%%%%%%%%%%%%%%%%%%%%%%%%%%%%%%%%%%%%%%%%%%%%%%%%%%%%%%%%%%%%%%%%%%%%%%%%%%%%%%%%%%%%%%%%%%%%%%%%%%%%%%%%%%%%%%%%%%
%%%%%%%%%%%%%%%%%%%%%%%%%%%%%%%%%%%%%%%

\section{Introduction}
%%%%%%%%%%%%%%%%%%%%%%%%%%%%%%%%%%%%%%%%%%%%%%%%%%%%%%%%%%%%%%%%%%%%%%%%%%%%%%%%%%%%%%%%%%%%%%%%%%%%%%%%%%%%%%%%%%%%%%%%
% general introduction 
%%%%%%%%%%%%%%%%%%%%%%%%%%%%%%%%%%%%%%%%%%%%%%%%%%%%%%%%%%%%%%%%%%%%%%%%%%%%%%%%%%%%%%%%%%%%%%%%%%%%%%%%%%%%%%%%%%%%%%%%

In wireless networks with battery-operated devices, energy saving mechanisms are of paramount importance in order to maximize %the increase 
network lifetime. 
Measurements have shown that the power consumption associated to the reception of packets is of the same order of magnitude as of that involved in the packets transmission %of packets %
%To this respect, an early discover 
~\cite{Feeney:infocom01,SMAC,WiseMAC}. % was that transmitting power  consumption is of the same magnitude of that required by receiving operations. 
Even worse, due to the power drainage 
operated by the radio frequency (RF) amplifier, the energy expenditure for passive communication operations -- e.g., 
overhearing the channel for collision avoidance mechanisms -- has a very high energy cost as well. %Hence, 
As a consequence, power saving mechanisms %techniques
 typically force low-power consumption modes for the RF interface, which in turn inhibits the communications capabilities (transmission, reception and channel sensing) of wireless devices for a certain fraction of time. The typical statement to this respect is that nodes %mobiles 
are put into {\em sleep} mode.  

A customary solution to increase network lifetime is therefore to periodically put nodes into sleep mode. This corresponds to introducing a duty cycle, intertwining sleep and active periods at the device level. 
%the definition of duty cycles of given period: wireless 
%perform a {\em periodic} wake-up cycle composed of a {\em sleep} period and an {\em awake} period. 
%During a 
%duty cycle it is then possible to reduce the power consumption by either %leveraging doze modes \textcolor{red}{DM:doze mode not defined} or to switch 
%off the RF amplifiers during the sleep period. We observe that this technique is indeed ubiquitous in 
Such approach is already in use in a variety of wireless technologies, 
%wireless networking, 
including power-save modes and beaconing techniques in IEEE802.11 %\textcolor{red}{DM: forget PCF, no one is using it, look at 802.11p instead} 
~\cite{Tseng} and wake-up cycling methods in wireless sensor networks%, as done, e.g., in the BMAC protocol
~\cite{BMAC}.%, just  to cite two major application cases. 

The %However, the 
introduction of wake-up scheduling has been identified in the past as a major source of performance degradation in wireless networks.
%, 
% both in ad-hoc networks literature and in sensor networks literature. The server vacation 
Nodes going in sleep mode increase indeed the latency associated to the delivery of messages. This has been subject to a number of 
research studies, in particular %in the wireless sensor networks domain, where it has been shown that the introduction of 
showing that 
duty cycles may lead to a data forwarding interruption problem in wireless sensor networks~\cite{SMAC,DMAC}. 
%introduced by the wake-up schedule increases end-to-end latency: the detrimental effect of duty 
%cycles on latency is showed in 
%~\cite{SMAC,DMAC}, where a data forwarding interruption problem is indentified. 
%The relation between service times at intermediate relays and end-to-end delay  \cite{more,more} in turn provides a clear trade-off between the duration of sleep periods, and the end-to-end delay. (\textcolor{red}{Here I would like to make a quick link to Chlamtac paper}) 

%In particular, the maximum rate of the duty cycle is then bounded by the request to satisfy end to end constraints. Conversely, battery exhaustion  requires that at each node the duty cycle is large in order to prolong the network lifetime. 
The problem of scheduling wake-up times may appear at first sight a standard optimization problem: one tries to maximize the wake-up periods while satisfying
%and still satisfy 
a given delay bound~\cite{GuhaGW10,LuSKG05} in order to minimize the associated energy consumption. 
However, in this trade-off there exists a further element of complexity since the scheduling of activity periods 
requires coordination at the network level. Indeed, nodes in sleep mode are not able to receive incoming messages, 
so that it is necessary that both transmitting and receiving nodes become active at the same time. This problem is 
often referred to as scheduling of rendezvous points among neighboring nodes. 
% since wake-up times when neighboring nodes 
% can communicate, namely {\em rendez-vouz}, require coordination. 
While it can be regarded as a synchronization issue, it is inherently different from clock synchronization~\cite{Li_Rus,Mirollo90}, 
where one typically aims at synchronizing devices' clocks network-wide using minimal local message exchange. 
%%%%%%%%%%%%%%%%%%%%%%%%%%%%%%%%%%%%%%%%%%%%%%%%%%%%%%%%%%%%%%%%%%%%%%%%%%%%%%%%%%%%%%%%%%%%%%%%%%%%%%%%%%%%%%%%%%%%%%%%
% unsolved problems 
%%%%%%%%%%%%%%%%%%%%%%%%%%%%%%%%%%%%%%%%%%%%%%%%%%%%%%%%%%%%%%%%%%%%%%%%%%%%%%%%%%%%%%%%%%%%%%%%%%%%%%%%%%%%%%%%%%%%%%%%

A recent research line %literature 
has addressed the design of quorum-based protocols for the power saving duty cycles% and issues related to their coordination
~\cite{Tseng,Zheng,Wu_aaa,Wu_powawa,Jiang_quorum}. %Several constraints have been identified 
%for the design of a periodic 
In general, wake-up schedules should be robust with respect to loss of synchronization due to clock skew and should account for possibly different energy/delay constrains on various nodes (due to, e.g., different types of battery installed in different devices etc.). Furthermore, protocols for determining the wake-up schedules should be able to adapt to changes in network topology, due to, e.g., node failures or addition of new nodes to the network. In the aforementioned works, the target is the definition of algorithms and protocols able to ensure that two neighboring nodes are able to communicate (i.e., both are active at the same time) at least once within a given (bounded) time-frame. A different approach to reduce energy consumption is adaptive listening~\cite{LuSKG05,BasuBroad06,GuhaGW10}, whereby nodes enter a low-power mode if no activity is detected on the channel for a given time interval.  However, in this approach, no guarantees can be given in terms of delay and/or energy constraints. 

%Wake-up scheduling should support asynchronous operations, be adaptive with minimum amount of 
%required signaling; also, asymmetry of wake-up schedules is mandatory when nodes have
% different power consumption and communication requirements.
%The above solutions target {\em neighbor discovery} time, i.e., the maximum allowed time for two neighboring nodes to detect each other, and then begin data communication; the design of rendezvous times is thus addressed mostly at node level. 
%In the literature, a few works address the fundamental issue of network performance under power saving duty 
%cycles under adaptive listening~\cite{LuSKG05,BasuBroad06,GuhaGW10}. 

In this work we take a different perspective and aim to answer to the following questions: {\em Assume that energy and delay constraints have been assigned network-wide: do feasible duty-cycle schedules exist?}% that are compliant with such constraints, and if 
{\em If so, what is the related 
energy-delay trade-off? Furthermore, do algorithms exist that achieve network-wide distributed scheduling 
of duty cycles, and scale in the number of required messages?}

Our main contributions can be summarized as follows: 
\begin{itemize}
\item a simple algebraic characterization of the wake-up scheduling problem; 
\item the analysis of the complexity of the wake-up scheduling problem: it is proved equivalent to that of integer factorization; %e problem of factorization of integer numbers.
\item heuristics and algorithms that solve the wake-up scheduling problem with message complexity linear in the number of links in the network;
\item fully distributed protocols implementing the aforementioned algorithms and test them in a system-level simulator to evaluate their performance.
% that are linear in the number of required messages and can be run in distributed fashion.
\end{itemize}

%%%%%%%%%%%%%%%%%%%%%%%%%%%%%%%%%%%%%%%%%%%%%%%%%%%%%%%%%%%%%%%%%%%%%%%%%%%%%%%%%%%%%%%%%%%%%%%%%%%%%%%%%%%%%%%%%%%%%%%%
% paper structure
%%%%%%%%%%%%%%%%%%%%%%%%%%%%%%%%%%%%%%%%%%%%%%%%%%%%%%%%%%%%%%%%%%%%%%%%%%%%%%%%%%%%%%%%%%%%%%%%%%%%%%%%%%%%%%%%%%%%%%%%

The remainder of the paper is organized as follows. Sec.~\ref{sec:related} reports on related works in the literature. Sec.~\ref{sec:model} introduces the system model and formally define the problem of wake-up scheduling problem. Then, we analyze in Sec.~\ref{sec:complexity} the complexity of the aforementioned problem by relating it to integer factorization while Sec.~\ref{sec:algos} introduces algorithms solving the problem with low messages overhead. In Sec.~\ref{sec:numerical} we evaluate the performance of the 
proposed protocols, obtained by means of system-level simulations. Finally, Sec.~\ref{sec:conclusions} is devoted to final a
conclusions and pointers to promising research directions.

%%%%%%%%%%%%%%%%%%%%%%%%%%%%%%%%%%%%%%%
%%%%%%%%%%%%%%%%%%%%%%%%%%%%%%%%%%%%%%%%%%%%%%%%%%%%%%%%%%%%%%%%%%%%%%%%%%%%%%%%%%%%%%%%%%%%%%%%%%%%%%%%%%%%%%%%%%%%%%%%

%%%%%%%%%%%%%%%%%%%%%%%%%%%%%%%%%%%%%%%%%%%%%%%%%%%%%%%%%%%%%%%%%%%%%%%%%%%%%%%%%%%%%%%%%%%%%%%%%%%%%%%%%%%%%%%%%%%%%%%%
%%%%%%%%%%%%%%%%%%%%%%%%%%%%%%%%%%%%%%%

\section{Related Work}\label{sec:related}
In wireless networks, periodic wake-up is a convenient mean to avoid idle listening to 
the channel and to prolong nodes'
lifetime~\cite{BMAC,WiseMAC,SMAC,RMAC}. Typically, the duty cycle
(defined as the fraction of time a node is in the active state) is set
to values of %
%The duty cycle is showed typically on 
the order
of some percent~\cite{SMAC,RMAC}. This is feasible when the data rate is low, i.e., 
on the order few packets per minute, as in low-rate wireless sensor networks.
%, which is made possible by the small data rates involved (some 
%packets per minute). 
In order to schedule periodic wake-ups, a simple %naive 
solution, adopted by 
BMAC and other WSN protocols, is to employ preamble sampling~\cite{BMAC,WiseMAC}.
Basically, the preamble is set long enough to guarantee node discovery at each 
wake-up. 

In order to avoid inefficiencies, it has been proposed to match the RF interface 
activation, i.e., the sleep period, to the traffic pattern. This is
the case of SMAC or TMAC protocols~\cite{SMAC,TMAC,RMAC}. In particular 
SMAC synchronizes the duty cycle by coordinating neighbouring sensors to the same 
slot to reduce the time spent listening to the channel. Adaptive change of the duty cycle is
supported in TMAC~\cite{TMAC}, whereas \cite{RMAC} matches duty cycles
to the dynamics of the routing protocol.
%routing on top. 
Tuning such protocols network-wide is typically not considered; %overlooked 
scalability issues do indeed arise, due to the need of a potentially
large %as % to lack of
%                                % scalability because 
%this may require a potentially %scheduling is operated using centralized solutions that 
%may require a 
%large 
amount of signaling messages. This is the core issue that we address 
in this paper: we design distributed algorithms for wake-up scheduling that
%which require 
present {\em linear} message complexity in terms of the network size {\em and} 
can operate local adaptation in case of nodes joining/leaving in a
fully distributed  
fashion.% \footnote{Despite the use case we refer to is the one of
        % sensor networks, it is worth noticing that the technique that we propose extends to any wireless network.}.
%(\textcolor{red}{Recall to motivate this sentence later on.})

In ad hoc networks literature, centralized planning of wake-up schedules %cycles 
was addressed in several works. Solutions that are close to the ones 
proposed in this paper can be found in~\cite{Tseng,Zheng,Wu_aaa,Wu_powawa,Jiang_quorum}. 
In~\cite{Tseng} the target is the design of power saving protocols for
WiFi networks,  
using {\em quorums} design and leveraging the power-save mode provided by the IEEE 802.11
standard. In \cite{Zheng} methods for designing asynchronous and
heterogeneous wake-up schedules are presented, %target heterogeneous wake-up cycle 
%lengths, 
where the heterogeneity may come from either %is motivated by different 
different applications running or by different node-level
features. %and network level requirements of nodes. 
Along the same 
line,~\cite{Wu_aaa} and~\cite{Chou} provide solutions for 
creating wake-up schedules whereby the ability of neighbouring nodes to communicate
(i.e., to be active in the same time interval) within a wake-up cycle % bounded time
%delay 
is ensured by combinatorial means. 
%with no need for a unique %target beacon transmission time. 
The complexity of quorum design limits the 
applicability of such combinatorial technique to regular or simple topologies;
%match networks; 
further 
insight into the design of quorums is found in \cite{Lai_OPODIS}. 

Some other works address the issue of network performance in the
presence of duty cycles% papers tackle the network performance of wake-up cycles: two 
%such papers are
~\cite{LuSKG05,%Lu2007,
GuhaGW10}. The goal of such works %aim there 
is to maximize the throughput and minimize the delay, while meeting a
given average power consumption constraint. It is shown that the
problem is, in general, 
%The problem addressed there is showed 
NP-complete; an analysis is provided for regular topologies (line,
grid and tree networks in~\cite{GuhaGW10}, tree and ring networks
in~\cite{LuSKG05}). The distributed heuristic provided
in~\cite{LuSKG05} does not provide guarantees in terms of delay
bounds. Centralized approaches only are considered in~\cite{GuhaGW10}.
%and heuristic solutions are proposed.  
%Within the context of those works, adaptive
%wake-up is considered, i.e., nodes wake up and wait neighbors' wake-up. We address 
%synchronization of duty cycles using rendezvous times.  

%Finally, i
In~\cite{Lu2007} the authors address %tackle also 
the joint design of 
routing and wake-up scheduling protocols. %; as it will be showed in Sec.~\ref{sec:algos}, 
%the algorithms proposed in our paper can be very easily coupled to routing algorithms. 
The idea is that routing should be aware of the wake-up schedules and
vice versa. A distributed implementation (based on the use of
distributed versions of Bellman-Ford algorithm) is also possible. The
need to coordinate routing and wake-up scheduling may create issues
related to the amount of signaling messages to be exchanged in
rapidly changing topologies. 

\subsection*{Main contributions} 
In this work we expose the equivalence of the wake-up scheduling problem 
to integer factorization. We observe that this follows naturally as soon as 
one imposes constraints on both delay and energy consumptions. Using basic 
algebraic tools~\cite{koblitz} we provide a framework that relies on a very 
limited set of input parameters, providing means to map energy consumption 
constraints and end-to-end delay requirements to wake-up scheduling. 

We further provide fully distributed heuristic algorithms that 
solve the wake-up scheduling problems with a message complexity $O(M)$
messages, where $M$ is the number of links in the network. The
proposed schemes are validated numerically using system-level 
simulations.

%%%%%%%%%%%%%%%%%%%%%%%%%%%%%%%%%%%%%%%
%%%%%%%%%%%%%%%%%%%%%%%%%%%%%%%%%%%%%%%%%%%%%%%%%%%%%%%%%%%%%%%%%%%%%%%%%%%%%%%%%%%%%%%%%%%%%%%%%%%%%%%%%%%%%%%%%%%%%%%%

%%%%%%%%%%%%%%%%%%%%%%%%%%%%%%%%%%%%%%%%%%%%%%%%%%%%%%%%%%%%%%%%%%%%%%%%%%%%%%%%%%%%%%%%%%%%%%%%%%%%%%%%%%%%%%%%%%%%%%%%
%%%%%%%%%%%%%%%%%%%%%%%%%%%%%%%%%%%%%%%

\section{System Model}\label{sec:model}
\begin{table}[t]\caption{Notation used throughout the paper}
\centering
\begin{tabular}{|p{0.10\columnwidth}|p{0.8\columnwidth}|}
\hline
{\it Symbol} & {\it Definition}\\
\hline
$V$,$N$ & set of nodes, number of nodes\\
$E$,$M$ & set of links, number of links\\
$d_i$ & degree of node $i$\\
$\Ng{i}$ & neighbors of node $i$\\
$i_j$ & $j$-th neighbor of node $i$\\
$T_0$ & time slot duration\\
%DM: removed, confusing$\Z$ & network time axis ($T_0=1$)\\
$n_i$ & wake-up period at node $i$\\
$\alpha_i$ & start epoch of wake-up cycle at node $i$\\ %DM 29/7time axis origin for node $i$\\
$a_{ij}$ & phase of node $i$ towards node $j$\\
$A_i$ & set of active slots for node $i$ within a period \\%phases at node $i$ \\
$x^{n}_{ij}$ & $n$--th rendezvous time for nodes $i$ and $j$\\
$(r,s)$& greatest common divider of $r,s\in \Z$\\
$[r,s]$& least common multiplier of $r,s\in \Z$\\
\hline
\end{tabular}
\label{tab:notation}
\end{table}

We consider a multi-hop wireless network represented as a graph $G(V,E)$ where $G$ is a set of 
$N$ vertices corresponding to nodes $j=1,\ldots,N$, and $E$ is a set of $M$ edges. Edge $ij$ 
connects vertices $j$ and $i$ if the two nodes are within mutual communication range. We 
assume a discrete time model and denote by $T_0$ the time slot duration. 

In the following we will assume $T_0=1$ for notation's sake. We consider 
a common time axis, and denote the origin by $0$. Nodes do not need to be 
aware of such common time axis: we will use it only for presentation purposes (see 
Sec.\ref{sec:sub:local}).

$\Ng{i}$ is the set of neighbors of node $i$; we further
assume that nodes are synchronized on time slot boundaries.

We now define a periodic wake-up schedule. To illustrate the concept
and notation, we refer to the simple example reported in Fig.~\ref{fig:concept}. 
Two neighboring nodes, namely node $1$ and node $2$, wake-up
with {\em period} $n_1=5$ and $n_2=3$, respectively. This means
that they wake every $n_1$ (respectively: $n_2$) slots. However, their
wake-up cycles can start from a different ``origin''. In the example,
node $1$ period starts at slot $1$ and node $2$ period at slot $2$,
respectively. In particular, let $\alpha_i$ be the local time 
origin for node $i$: in the example, the two nodes wake up at time 
$0$ with respect to their local time, but at time $1$ and $2$ 
with respect to the common time. 

We want to characterize the conditions under which for every 
pair of nodes in the network, a wake-up schedule provides a set of times when nodes are both awake 
and can communicate; those times are called also {\em rendezvous} times. 
 
Nodes may have {\it multiple} active slots within a wake-up period. We denote by 
$A_i$ the set of such active slots for node $i$. In other words, node $i$ follows an activation
cycle with period $n_i$, whereby the active slots are characterized by
the set $A_i$. 
We call an element of $A_i$ a {\it phase}.

Formally, we define a periodic {\em wake-up schedule for node $i$} as a function $S(i)=(A_i,n_i)$,
 where $(A_i,n_i)$ is a pair given by the wake-up epoch set
 $A_i\subseteq \{0,1,\ldots,n_i-1\}$  %%DM 29/7 \Z$ 
and the wake-up period $n_i\in \Z$. Under a wake-up schedule, node $i$ will wake up at times
\[
A_i+n_i \Z :=\{a+k \cdot n_i,k \in \Z,a \in A_i \}
\]
We say that node $i$ performs a periodic wake-up with period $n_i$ and
epochs in $A_i$. With a slight abuse of terminology, we will call a
{\it wake-up schedule} the set of functions $S=\{S_i\}_{i=1,\ldots,N}$; also, 
 for ease of reading, we will denote $a_{ij}\in A_i$ the phase that node $i$ uses to communicate 
with node $j$.

Denote $X=\{x_{ij}^{n}\}$ the set of the {\em meeting times (or: rendezvous points)} for nodes 
$i$ and $j$, i.e., the set of time slots when node $i$ and $j$ are both active: %awake: 
 $X=\left(A_i+n_i \Z \right) \cap\left( A_j+n_j \Z \right)$. As it will be clear 
in the following, only two cases are possible: $X=\emptyset$ or $|X|=\infty$, i.e., 
either there are no meeting times or there is an infinite number thereof. Quest for {\em periodicity}
 is the reason for such dichotomy.

 \begin{figure}[t]
\centering
\includegraphics[width=9cm]{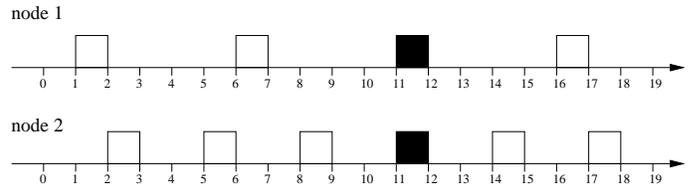}%
\caption{Wake-up schedule example: $\alpha_{1}=1$, $n_1=5$, $\alpha_{2}=2$  and $n_2=3$; 
nodes meet every $15$ slots starting at slot $11$.}\label{fig:concept}
%\caption{Wake-up synchronization: node $1$ and $2$ wake up schedule when $a_{11}=1$, $n_1=5$, $a_{22}=2$  and $n_2=3$; 
%nodes meet every $15$ slots starting at slot $11$.}\label{fig:concept}
\end{figure}

%%reversing the reasoning; first lower bounds (energy), then upper
%%bounds (latency)
We consider first energy constraints. Assuming that the energy
expenditure is proportional to the fraction of time a node is in the
active state, one wants to prolong the lifetime of node $i$ by letting
the {\em duty cycle} $|A_i|/n_i$ satisfy constraints of the form:
\begin{equation}\label{eq:constr_lower}
%DM 29/7- reversed for readability purposes |A_i|/n_i\leq  1/L_i \; i=1,\ldots,N
%1/L_i \geq |A_i|/n_i  \; i=1,\ldots,N
n_i \geq L_i \cdot |A_i|, \; i=1,\ldots,N
\end{equation}
We consider the possibility of nodes having different energy
constraints (i.e., different values for $L_i$). This may reflect,
e.g., heterogeneity in the battery capacity. 

Further, depending on the specific applications running on top of the 
network, specific requirements for the end-to-end delay traffic may 
have to be satisfied. For instance, the delay bound may represent the 
maximum latency that is allowed for a certain sensor reading.

%In turn, such bounds will be mapped \ref{app:boundsdym} at each node
%for the wake-up duty-cycle. I.e., we 
In the remainder of the paper, we assume that constraints of such a type
are given in
terms of the maximum delay between two subsequent rendezvous between
node $i$ and {\it any} of its neighbors $j$. We denote such
constraint by $U_i$; the constraint can then be written as: 
\begin{equation}\label{eq:constr_upper}
x_{ij}^{n+1}-x_{ij}^{n} \leq U_{i}, \; \forall j\in N_g(i).%1,\ldots,N
\end{equation}

It is important to notice here local constraints $L_i$ and $U_i$ are input 
to the problem. In App.~\ref{app:boundsdym}, we show how end-to-end delay constraints can be
turned into constraints of the form (\ref{eq:constr_upper}), while
accounting at the same time for the presence of energy constraints as 
in (\ref{eq:constr_lower}). 

%assume that for a given $U_i$. Assume that such numbers have been determined, the constraint on rendezvouz times writes
%%DM: removed, sentence confusing
%Furthermore, there might exist dependence 
%on the index $i$ for such constraints since different nodes may have different requirements in terms of 
%communication and energy expenditure. 
\begin{defn}
A wake-up schedule $S$ is said to be {\it feasible} if for
every node $i$, the number of rendezvous points with node $j$, where $j \in N_g(i)$ 
is infinite and if it respects the energy constraints (\ref{eq:constr_lower}).
\end{defn}
\begin{defn}
A wake-up schedule is said to be {\em tight} if it is feasible and if
it further respects 
%obeys to 
the constraints (\ref{eq:constr_upper}) .
%and  (\ref{eq:constr_lower}) and  for
%every node $i$, the number of rendezvous points with node $j$, where $j \in N_g(i)$, is infinite. 
\end{defn}
%a solution exists 
%\begin{equation}\label{eq:sol}
%x_{ij} = a_{ij} \mod n_i = a_{ji} \mod n_j = x_{ji} 
%\end{equation} 
%for $ij \in E$. 
\begin{defn}
A wake-up schedule is said to be  {\em strictly feasible
  (respectively: tight)} if it is feasible (respectively: tight)
and $|A_i|=1$ for all $i=1,\ldots,N$. 
\end{defn}
The wake-up scheduling problem can be formalized as: %problem that we want to solve is 
\begin{prob}[{\tt Strong WAKE-UP}]\label{prob:wakeup_strong}
Given graph $G=(V,E)$ and constraints (\ref{eq:constr_upper}) and  (\ref{eq:constr_lower}), determine 
a (strictly) tight wake-up schedule.
\end{prob}
%We say that the  WAKE-UP problem is solvable if there exists a
%feasible wakeup scheduling. 
A trivial necessary condition is that $U_i\geq L+1=1+ \max \{L_{j_1},\ldots,L_{j_n}\}, \quad
\forall ij_k \in E$, which we assume holds in the remainder of the paper.%eople.
%following we assume.

If a wake-up schedule is feasible but not tight, some of the constraints (\ref{eq:constr_upper}) are not satisfied. 
This origins a weaker problem that is 
\begin{prob}[{\tt Weak WAKE-UP}]\label{prob:wakeup_weak}
%WAKE-UP (Weak): 
Given graph $G=(V,E)$ and constraints (\ref{eq:constr_upper}) and  (\ref{eq:constr_lower}), determine 
a (strictly) feasible wake-up schedule.% that minimize the number of violations.
\end{prob}

We denote as {\em violation} a condition such that constraints
(\ref{eq:constr_upper}) are not satisfied and yet a wake-up schedule is feasible. 

\subsection{Key Assumptions}
We summarize here the main assumptions that we will use in the
remainder of the paper:
%, we use the following assumptions:
\begin{itemize}
\item Nodes do not have global knowledge of the common time axis;
\item Nodes do not have global knowledge of the $L_i$ and $U_i$
values assigned to other nodes in the network;
%Before proceeding further, there are a few key assumptions in our
%model to be underlined. First, we 
%\item The network is lightly loaded, so that all messages queued at 
%each node can be transmitted within one slot; 
 %the network is lightly loaded, so that we can neglect queueing delays at intermediate nodes. Observe
 %that {\em more to send} flags can be used in case more packets are ready in queue; the effect it to
 %override wake-up cycles in case of heavy load. 
\item Energy consumption is directly proportional to the duty cycle. This implies that 
we neglect non-linear effects and energy costs associated to wake-up operations;
\item Neighboring nodes are synchronized on common time slot boundaries \cite{SMAC}.
% performed at each cycle. In order to perform wake-up scheduling, we assume that synchronization 
%of time-slots exist; techniques exist to achieve synchronization on the order of microseconds.
%However, we only require neighboring nodes to be synchronized on common time slot boundaries.
\end{itemize}
We remark that our results apply to general networks; however, due to the customary need to prolong battery 
lifetime, wireless sensor networks will be used as our reference case throughout the paper.

%%%%%%%%%%%%%%%%%%%%%%%%%%%%%%%%%%%%%%%
%%%%%%%%%%%%%%%%%%%%%%%%%%%%%%%%%%%%%%%%%%%%%%%%%%%%%%%%%%%%%%%%%%%%%%%%%%%%%%%%%%%%%%%%%%%%%%%%%%%%%%%%%%%%%%%%%%%%%%%%

%%%%%%%%%%%%%%%%%%%%%%%%%%%%%%%%%%%%%%%%%%%%%%%%%%%%%%%%%%%%%%%%%%%%%%%%%%%%%%%%%%%%%%%%%%%%%%%%%%%%%%%%%%%%%%%%%%%%%%%%
%%%%%%%%%%%%%%%%%%%%%%%%%%%%%%%%%%%%%%%

\section{Characterization of the Wake--Up Problem}\label{sec:complexity}
In this section we provide an algebraic characterization of the WAKE-UP problem. % based on a few elementary tools. 
Using standard notation, we denote by $(r,s)$ and by $[r,s]$ the
greatest common divider (gcd) and the least common multiplier (lcm) of $r,s \in \Z$, 
respectively\footnote{We use notation $[r,s]$ to also denote the interval of the integers $u$ 
such that $r\leq u \leq s$; the meaning of the notation is clear from the context.}. Integer $r$ divides $s$, which we write
as $r|s$, if $s=r\cdot v$ for some integer $v$.

We first recall a fundamental result:
%that is the basis for the rest of the discussion, namely the Chinese Remainder theorem. 
\begin{thm} (Chinese Remainder theorem \cite{koblitz})
Let $n_1,n_2,a_1,a_2 \in \mathbb{Z}$, consider the system 
\begin{equation}\label{eq:chirem}
\begin{cases}
x= a_1 \mod n_1 \\
x= a_2 \mod n_2
\end{cases}  
\end{equation}
Then, 
\begin{enumerate}
\item[(i)] system (\ref{eq:chirem}) has a solution if and only if $(n_1,n_2)|a_1-a_2$;
%\item[(ii)] a solution $x_0$ of (\ref{eq:chirem}) can be obtained using the extended Euclidean algorithm;
\item[(iii)] the general solution of (\ref{eq:chirem})  is given by $x=x_0 \mod [n_1,n_2]$, where $x_0$ is 
a particular solution of (\ref{eq:chirem}).
\end{enumerate}
\end{thm}

In order to derive a particular solution $x_0$ at step (ii), the extended Euclidean algorithm~\cite{koblitz} 
guarantees that integer $u_1,u_2$ exist that solve the following B\'ezout's identity 
\[
 n_1 u_1  - n_2 u_2 = (n_1,n_2)
\]
from which 
\[
x_0=a_1-n_1 \frac{u_1 (a_1-a_2) }{(n_1,n_2)}=a_2 - n_2 \frac{u_2 (a_1-a_2) }{(n_1,n_2)}
\]
The extended Euclidean algorithm has complexity $O(\log^3 \max \{n_1,n_2\})=O(1)$. It is clear that when system 
(\ref{eq:chirem}) consists of $k>2$ congruences, it can be solved 
iteratively; the solution can be determined repeating the operation for a single pair 
$k-1$ times. The time complexity for one step is $O(1)$, which in turn becomes on 
the order of $O(k)$. 

From the model introduced in Sec.~\ref{sec:model}, the solution 
of the {\tt WAKE-UP} problem has to satisfy the following relation:%\vspace{-2.5ex}
\begin{equation}\label{eq:sol}
a_{ij} + h \cdot n_i = a_{ji} + k \cdot n_j,
%\mod n_i = a_{ji} \mod n_j
\end{equation} %\vskip-2ex
for $ij \in E$ and $a_{ij}\in A_i$ and $a_{ji}\in A_j$, where $h,k \in
\mathbb{Z}$. The sequence of rendezvous times
$\{x_{ij}^n\}_{n\in\mathbb{Z}}$ can therefore be seen as solution of
(\ref{eq:chirem}), with $a_1=a_{ij}$, $n_1=n_i$, $a_2=a_{ji}$ and $n_2=n_j$.
%, 
%i.e., $x_{ij}^k=a_{ij} \mod n_i= a_{ji} \mod n_j$. %\vspace{-2.5ex}
% From the model introduced in Sec.~\ref{sec:model}, the solution of the WAKE-UP problem has to satisfy the following relation:
% \begin{equation}\label{eq:sol}
% a_{ij} \mod n_i = a_{ji} \mod n_j
% \end{equation} 
% for $ij \in E$ and $a_{ij}\in A_i$ and $a_{ji}\in A_j$. In particular, rendezvous times $\{x_{ij}^k\}_{k\in\mathbb{Z}}$ are solutions of (\ref{eq:sol}), 
% i.e., $x_{ij}^k=a_{ij} \mod n_i= a_{ji} \mod n_j$. 

Existence of feasible schedules follows as a consequence of the Chinese Remainder theorem. % if one relaxes the
% wake-up problem constraints by not considering constraints of the form
% (\ref{eq:constr_upper}), the problem can be solved %is {\em always}
%                                 %feasible with 
% by means of a distributed algorithm having linear complexity in the 
% size of the network. 

\begin{thm}\label{cor:nobound}
%Let $U_i=\infty$, 
%There always exists a strictly feasible wake-up schedule. Further, it
%is possible to construct it in time $O(N)$ 
%with a fully distributed asynchronous algorithm requiring $O(M)$ messages. 
A strictly feasible wake-up schedule always exists, and can be
constructed 
%There always exists a strictly feasible wake-up schedule. Further, it
%is possible to construct it 
in time $O(N)$ 
with a fully distributed asynchronous algorithm requiring $O(M)$ messages. 
\end{thm}
%\begin{IEEEproof}
\proof
A feasible wake up schedule is constructed in $N$ steps choosing for all nodes $n_i \geq L_i$ 
such that $(n_i,n_j)=1$ for all pairs neighboring nodes $i$ and $j$. For each edge we can 
apply the Chinese Reminder Theorem to its end nodes. Irrespective of $a_i$ and $a_j$, 
it holds $1=(n_i,n_j)|(\alpha_i-\alpha_j)$ at each step. The algorithm 
attains by construction a feasible wake-up schedule. 
\endproof
%\end{IEEEproof}
%\begin{rem}
%Notice that full synchronization has complexity $\Theta(|V|^2)$ and requires $\Theta(NM)$ messages.
%\end{rem}

% The above construction shows a a basic effect: the request of synchronization of duty cycles 
% introduces a {\em delay drift} (DD) rooted in the Chinese Remainder theorem. In the worst case, i.e., 
%  $n_i$ and $n_j$ coprime for all $i\not =j$, $[n_i,n_j]=n_i\cdot n_j$, and products become exponentially
% large with the network diameter. In order to minimize such drift, one would request $[n_i,n_j]=\min\{n_i,n_j\}$ 
% at each step, i.e., either $n_i|n_j$ or vice versa. In what follows we show that such request 
% is difficult to obtain.
\begin{rem}
The above construction shows a basic effect: the request of synchronization of duty cycles 
introduces a {\em delay drift} (DD). In the worst case since
$[n_i,n_j]=n_i\cdot n_j$, products become exponentially large with respect to the network diameter. 
In order to minimize such drift, one would request $[n_i,n_j]=\min\{n_i,n_j\}$ at each step, i.e., 
either $n_i|n_j$ or vice versa.%In what follows we show that such request is difficult to obtain.
\end{rem}

We get therefore:\vskip-3ex
\begin{cor}
The {\tt Weak WAKE-UP} problem can be solved in polynomial time.
\end{cor}\vskip-2ex

In the general case, when constraints on the maximum wake-up period are imposed, the problem 
becomes difficult due to the simultaneous requirements on the phases of nodes.  Given a certificate 
wake-up schedule, we can verify with $2E$ congruences if the constraints 
(\ref{eq:constr_upper}) and (\ref{eq:constr_lower}) are satisfied. Also, congruences appearing in
(\ref{eq:sol}) have the cost of computing $(n_i,n_j)$ for each pair of nodes  $i,j$, i.e., 
a solution of {\tt Strong WAKE-UP} problem can be verified in polynomial time with respect to the 
input size, i.e., $M$.

In the general case, when constraints on the maximum wake-up period are imposed, the given problem 
is made difficult due to the simultaneous requirement on the phases of nodes. It indeed 
belongs to the polynomial time verifiable algorithms class, i.e., {\tt WAKE-UP} is NP.  In fact, 
given a certificate wake-up schedule, we can verify with $2E$ congruences if the constraints 
(\ref{eq:constr_upper}) and (\ref{eq:constr_lower}) are satisfied. Also, congruences appearing 
(\ref{eq:sol}) have the cost of computing $(n_i,n_j)$ for each pair of nodes  $i,j$, i.e., 
a solution of {\tt WAKE-UP} problem can be verified %in polynomial time,
                                %i.e., 
with $O(M)$ operations.

% The wake-up scheduling problem relates to the following

% \begin{prob}[INTEGER FACTORIZATION]\label{prob:factorize}
% Given integer $R$, determine the factorization of $R$.
% \end{prob}
The wake-up scheduling problem relates to the {\tt INTEGER FACTORIZATION}
problem: 
%DM 28/12
%following \vspace{-2.5ex}
%\begin{prob}[{\tt INTEGER FACTORIZATION}]\label{prob:factorize}
Given integer $R$, determine the factorization of $R$.
%%remark: DM changed notation, 29/7 (N is for us the number of
%%nodes!). replaced N->R

Notice that while the answer to the question whether a factorization of $R$ exists -- namely the {\tt PRIMALITY} 
problem -- can be actually solved in polynomial time~\cite{PRIMEisP}, the {\tt INTEGER FACTORIZATION} problem is harder. At present there 
exist several algorithms for factoring an integer $R$ that have
estimated complexity either exponential in the size of the binary
representation of $R$ (as in the case of the exhaustive trial method),
or sub-exponential (as in the case of the continuous fraction
method~\cite{giblin}).
%, 
%i.e., in the form $O(\exp(\log^\frac1h n)))$, $h\geq 2$; for example $h=2$ for the continuous fraction method~\cite{giblin}. 
No polynomial time deterministic algorithm is known that solves {\tt INTEGER FACTORIZATION}~\cite{arora}.

\begin{algorithm}[t]
\caption{{\tt PERIOD$(L_i,U_i)$}}
\label{alg:algorit0}
\begin{algorithmic}[1]
\REQUIRE {$L_i \leq U_i$ $\rightarrow Set:$ $x=L_i,n_i=L_i$}
%\STATE {Select: $n_i=L_i$}
\WHILE{$x\leq U_i$} 
\IF{$B | x$}
\STATE $n_i \leftarrow x$; BREAK
%\STATE BREAK
\ENDIF
\STATE $x\leftarrow x+1$ 
\ENDWHILE 
\end{algorithmic}
\end{algorithm}

\begin{thm}
{\tt WAKE-UP} is as hard as {\tt INTEGER FACTORIZATION}.
\end{thm}
%\begin{IEEEproof}
\proof
Consider a polynomial time reduction to the {\tt INTEGER FACTORIZATION} problem. The input is a integer $R>0$. The reduced instance of 
the {\tt WAKE-UP} is a graph composed of two nodes, $1$ and $2$, and one edge that joins them; furthermore 
$U_2=L_2=R$, $U_1=R-1$ and $L_1=2$. Assume that a solution of the {\tt INTEGER FACTORIZATION} problem is given, i.e., 
$R$ has proper factors $m$ and $n$, i.e., $R=m\cdot n$, $m\not |1$, $n1 \not |1$. Consider $n_2=L_2=U_2=R$ 
and $n_1=m$. Furthermore, let $a_1=a_2=0$: then $(0,m)$ and $(0,R)$ are a strictly tight wake-up 
schedule. Vice versa, assume that a feasible schedule exists. Indeed $n_2=L_2=U_2=R$ by construction. Observe 
that the $n_1>1$ since $L_1\geq 2$. However, the Chinese Reminder Theorem ensures that $[n_1,R] \leq R$, 
 which in turn implies that $n_1$ and $R$ cannot be co-prime, otherwise $[n_1,R]=n_1 R \geq 2 R > R$, 
 a contradiction. Hence $n_1|R$, and it is a proper factor of $R$ since $1<2=L_2 \leq n_1\leq R-1 < R$. 
Hence, by solving the {\tt WAKE-UP} problem, we have answered to the problem of factorizing  
$R$ as $R=n_1 \cdot q$ for some pair of proper factors $n_1,q>1$, which concludes the proof.
%\end{IEEEproof}
\endproof

As reported in App.~\ref{sec:isomorphism}, there exists a further argument for the equivalence proved above: 
 a connection exists between {\tt WAKE-UP} and {\tt INTEGER FACTORIZATION} through a standard isomorphism. {\em Also, 
{\tt WAKE-UP} is not NP-complete unless NP = co-NP since it is known that %INTEGER 
{\tt INTEGER FACTORIZATION} is NP-complete if and only if NP = co-NP, which (at present) is believed 
 a convincing argument for ensuring that no polynomial time complexity algorithms exists~\cite{arora}.}

A question naturally arising is whether similar issues hold if we
formulate the problem in continuous time, as opposed to the
discrete-time framework. In App.~\ref{sec:continuous} %\cite{techrep} Sec.~\ref{sec:continuous} we will
we show that, under a rationality assumptions, the two formulations
are equivalent. 

Therefore, no efficient algorithm for exactly solving {\tt Strong WAKE-UP} can be
designed. In the next section we will introduce two heuristics for the
{\tt Strong WAKE-UP} which have message complexity linear in the number of links in
the graph.

\subsection{Local time versus global time}\label{sec:sub:local}

From the implementation standpoint, a remark is needed: according to the framework proposed, the 
elements of $A_i$ are defined with respect to the common time. But, one might question that 
nodes operations are likely based on local time only. However, from the Chinese 
Remainder theorem, phase differences $\mod (n_i,n_j)$ matter. In 
practice, nodes will know their relative phase, i.e., the time when a wake-up occurs
for a neighbor with respect to their own time axis. Thus, node $i$ will store the phase corresponding to node $j$ 
as $\Delta_{ij}=a_{ij}-\alpha_i$ and node $j$ will do the same for $\Delta_{ji}=a_{ji}-\alpha_j$. 

We can now see that phase differences with respect to the global time axis can be handled 
using values measured with respect to local time axes. In fact, 
% \begin{eqnarray}
% &&a_{ij}-a_{ji} \mod n_i=\nonumber \\
% &&(\Delta_{ij} - \Delta_{ji})\mod n_i - (\alpha_i-\alpha_j) \mod n_i \nonumber \\
% %&&(\Delta_{ij} - \Delta_{ji})\mod n_i - \alpha_j \mod n_i
% \end{eqnarray}
%The last term can be easily measured time-stamping the relative index of time slots at node $j$. Finally, 
\begin{eqnarray}
&&a_{ij}-a_{ji} \mod (n_i,n_j) \nonumber \\
&&=(\Delta_{ij} - \Delta_{ji})\!\! \mod (n_i,n_j)  - (\alpha_i-\alpha_j)\!\!\! \mod (n_i,n_j) \nonumber \\
&&=(\Delta_{ij} - \Delta_{ji})\!\! \mod (n_i,n_j)  + (\alpha_j-\alpha_i)\!\!\! \mod (n_i,n_j) \nonumber 
\end{eqnarray}
Where last term corresponds to the offset of the two local clocks and can be easily measured time-stamping 
the local index of a certain time slot. 

In the rest of the paper we will refer to the common time axis for the sake of clarity. 

%%%%%%%%%%%%%%%%%%%%%%%%%%%%%%%%%%%%%%%
%%%%%%%%%%%%%%%%%%%%%%%%%%%%%%%%%%%%%%%%%%%%%%%%%%%%%%%%%%%%%%%%%%%%%%%%%%%%%%%%%%%%%%%%%%%%%%%%%%%%%%%%%%%%%%%%%%%%%%%%

%%%%%%%%%%%%%%%%%%%%%%%%%%%%%%%%%%%%%%%%%%%%%%%%%%%%%%%%%%%%%%%%%%%%%%%%%%%%%%%%%%%%%%%%%%%%%%%%%%%%%%%%%%%%%%%%%%%%%%%%
%%%%%%%%%%%%%%%%%%%%%%%%%%%%%%%%%%%%%%%

\section{Algorithms Design}\label{sec:algos}
%We are interested in distributed solutions of the WAKE-UP problem that require $O(M)$ messages. 
We assume that each node $i$ maintains a sorted list of neighbors $\Ng{i}$ according to their degree $d_i$.

The algorithms that we adopt in the following use a certain factor basis $B=\{p_1,p_2,\ldots,p_n\}$ where $p_i$s are 
prime numbers \cite{koblitz}. For the sake of notation, let us say that $B|n$ if $n$ factorizes in $B$. % They choose 
% period $n_i$ such that it factorizes according to $B$ using the procedure described in Alg.~\ref{alg:algorit0}. 
Observe that the choice of the factors of the wake-up periods at each node is critical; the Chinese Remainder 
Theorem suggests that period $n_i$ should be chosen free of large factors.

% \begin{algorithm}[t]
% \caption{PERIOD$(L_i,U_i)$}
% \label{alg:algorit0}
% %\small
% \begin{algorithmic}
% \REQUIRE {$L_i \leq U_i$}
% \STATE {Select: $n_i=L_i$}
% \WHILE{$n_i\leq U_i$  AND $B\not | n_i$}
% \STATE $n_i\leftarrow n_i+1$ 
% \ENDWHILE 
% \end{algorithmic}
% \end{algorithm}

% Alg.~\ref{alg:algorit0} calculates PERIOD$(L_i,U_i)$ as the smallest integer that falls within interval $[L_i,U_i]$ 
% in the form $n_i=p_1^{q_1}\cdots p_n^{q_n}$, $0\leq q_i\in \Z$. %We observe that for numbers of practical use (i.e., $<10^9$), 
% efficient factorization algorithms do exist \cite{giblin}.
Nodes choose period $n_i\in [L_i,U_i]$ such that it factorizes according to $B$ using the procedure described 
in Alg.~\ref{alg:algorit0}. Alg.~\ref{alg:algorit0} returns the smallest integer in $[L_i,U_i]$ that can be written 
as $p_1^{q_1}\cdots p_n^{q_n}$, $0\leq q_i\in \Z$, provided that such a number exists. If it does not exist, $L_i$
is returned. %Efficient factorization algorithms do exist for numbers of
%practical use (i.e., $<10^9$)~\cite{giblin}.
 We now introduce two heuristics for {\tt Strong WAKE-UP}.

%%%%%%%%%%%%%%%%%%%%%%%%%%%%%%%%%%%%%%%%%%%%%%%%%%%%%%%%%%%%%%%%%%%%%%%%%%%%%%%%%%%%%%%%%%%%%%%%%%%%%%%%%%%%%
\subsection{{\tt BFS WAKE-UP} Algorithm}\label{sec:equential}
%%%%%%%%%%%%%%%%%%%%%%%%%%%%%%%%%%%%%%%%%%%%%%%%%%%%%%%%%%%%%%%%%%%%%%%%%%%%%%%%%%%%%%%%%%%%%%%%%%%%%%%%%%%%%

Once nodes have chosen their periods according to Alg.~\ref{alg:algorit0}, the simplest rationale is 
to align their phases pairwise. Alg.~\ref{alg:algorit1} is based on a breadth first 
search (BFS) visit of the network rooted at a given node; as customary for BFS search 
a queue $Q$ is employed.  

Each node $i$ maintains information about nodes in its neighborhood, i.e., the node ID, the phases 
set $A_{i}$ and a special flag $f$ that is used for the BSF procedure, where (i) $f(j)=0$ means $j$ 
unexplored (ii) $f(j)=1$ means node $j$ discovered but not all neighbors of $j$ examined (iii) $f(j)=2$ 
node $j$ explored. By 
default, node $i$ wakes up at the origin of its own period, i.e., at $\alpha_i$. 

%%%%%%%%%%%%%%%%%%%%%%%%%%%%%%%%%%%%%%%%%%%%%%%%%%%%%%%%%%%%%%%%%%%%%%%%%%%%%%%%%%%%%%%%%%%%%%%%%%%%%%%%%%%%%%%%%%%%%%%%

The algorithm selects all unexplored nodes and enqueues them at most once, so that it terminates in $N$ steps. 

\begin{algorithm}[t]
\caption{\tt BFS WAKE-UP}
\label{alg:algorit1}
%\small
\begin{algorithmic}[1]
\REQUIRE {$L_i \leq U_i$ $\rightarrow Default:$ $A_i=\{\alpha_i\}$} %\{Default phase set\}
%\STATE {\bf default} $A_i=\{\alpha_i\}$ \{Default phase set\}
\FORALL {$i \in V$}
%\STATE $n_i=\min\{ L_i\leq n_{i} \leq U_{i}, \; i,j=1,\ldots,N$ and $n_i|U$\} 
\STATE $n_i=\mbox{{\tt PERIOD}}(L_i,U_i)$
\STATE $f(i) \leftarrow 0$  %\COMMENT{Mark $i$ unexplored}
\ENDFOR
\ENSURE {Select: $i\in V$ with largest degree}
\STATE $f(i)=1$  
\STATE $Q\leftarrow \{i\}$      \COMMENT{Enqueue the first node}
\WHILE{$Q \not = \emptyset$ \COMMENT{Visit all the nodes in BFS}}
\STATE $i\leftarrow Q[1]$ \COMMENT{Get the node on top of the queue}
\FORALL {$j \in \Ng{i}$}
%\STATE $NG_i \leftarrow \Ng{i} $ \COMMENT{Enlist the neighboring nodes of $i$}
%\WHILE{$NG_i \not = \emptyset$ \COMMENT{Visit all neighbors}}  %
%\STATE $j=NG_i[1]$ \COMMENT{Pick the first neighbor}     
\IF{$f(j)=0$ \COMMENT{If it was not explored before}}
\STATE $Q \leftarrow \{j\}$  \COMMENT{Enqueue this node}
\STATE $\alpha_j\leftarrow \alpha_i$  \COMMENT{Align time axis of node $i$ and node $j$}  
\STATE $A_j\leftarrow \{\alpha_j\}$   \COMMENT{Update the phase if $j$} 
\STATE $x_{ij}^k=x_{ji}^k=\alpha_i+k\cdot [n_i,n_j],\quad k\in \mathbb Z$ \COMMENT{Rendezvous times}
\STATE $f(j)\leftarrow 1$ \COMMENT{Visited}          
\ENDIF
%\STATE    $NG_i = NG_i \setminus \{j\}$ \COMMENT{Delist $j$}
%\ENDWHILE   
\ENDFOR
\STATE    $Q = Q \setminus \{i\}$  \COMMENT{Dequeue $i$}
%\STATE    $n_i\leftarrow ([n_{i},n_{j}],\ldots,[n_{i},n_{k}])$ \COMMENT{Use the largest possible period at $i$}
\STATE    $n_i\leftarrow [n_{i},(n_{i_1},\ldots,n_{i_{d_i}})]$ \COMMENT{Use largest possible period at $i$}
\STATE    $f(i)=2$  \COMMENT{Explored}            
\ENDWHILE 
\end{algorithmic}
\end{algorithm}

%%%%%%%%%%%%%%%%%%%%%%%%%%%%%%%%%%%%%%%

Alg.~\ref{alg:algorit1} reports on the pseudocode of the centralized version of the {\tt BFS WAKE-UP} algorithm;  
however, the {\tt BFS WAKE-UP} algorithm can be easily implemented in a distributed fashion \cite{Cormen}. Observe that {\tt BFS WAKE-UP} 
attains phase synchronization and a strictly feasible schedule. 

 Below, we formalize the properties of the algorithm; the proof is reported in App.~\ref{sec:proofs}.
 %the proof is reported in \cite{techrep} for the sake of space.
\begin{thm}\label{thm:BFS}
{\tt BFS WAKE-UP} has the following properties:
\begin{enumerate}
\item[  i.] it produces a strictly feasible wake-up schedule; 
\item[ ii.] it has complexity $O(E)$; 
\item[iii.] a distributed implementation requires $O(M)$ messages.
\end{enumerate}
\end{thm}

\begin{rem}
 For implementation purposes, the above algorithm can be very easily coupled to the construction 
of shortest path trees \cite{Cormen}; in such a way customary algorithms such as Dijkstra 
can conveniently provide {\em joint routing and wake-up scheduling} with no need for 
separate message exchange; in turn, the weight of link $ij$ is represented naturally by $[n_i,n_j]$.
\end{rem}

The convergence time of the algorithm is proportional to the diameter of the network~\cite{Cormen};
 the {\tt BFS WAKE-UP} requires a preliminary leader election procedure; further its convergence 
may be slow. Accordingly, we devised an alternative heuristic algorithm, described below, which effectively addresses
such shortcomings by parallelising the computation of the scheduling.

%%%%%%%%%%%%%%%%%%%%%%%%%%%%%%%%%%%%%%%%%%%%%%%%%%%%%%%%%%%%%%%%%%%%%%%%%%%%%%%%%%%%%%%%%%%%%%%%%%%%%%%%%%%%%%%%%%%%%%%%
\subsection{{\tt PARALLEL WAKE-UP} Algorithm}\label{sec:parallel}
%%%%%%%%%%%%%%%%%%%%%%%%%%%%%%%%%%%%%%%%%%%%%%%%%%%%%%%%%%%%%%%%%%%%%%%%%%%%%%%%%%%%%%%%%%%%%%%%%%%%%%%%%%%%%%%%%%%%%%%%

The design of a parallel algorithm for wake-up synchronization requires careful handling of 
nodes phases; in particular, in order to avoid the increase of the size of the $A_i$s, it is 
possible to operate pairwise adaptation of nodes phases (phase adaptation).

\begin{lem}\label{prop:adjphase}
Let $ij,jk\in E$ and fix $n_i,n_j,n_k$, $a_{kj}$ and $a_{ij}$. A feasible wake-up schedule for $i,j,k$ 
such that $a_{jk}=a_{ji}$ exists if and only if
$(n_j,(n_i,n_k))|(a_{kj} - a_{ij})$. If such condition holds, the schedule is obtained for %\vspace{-2ex}
%\[
$a_{jk}=a_{ji} = a_{kj}+n_k \cdot (x_0\cdot e)$,  
%\]\vskip-3ex
where $e, x_0 \in \mathbb{Z}$ are solutions of %\vspace{-3ex}
%\[
$a_{kj}-a_{ij}= (n_j,n_k)\cdot e \cdot x_0 + (n_i,n_j)\cdot e
\cdot y_0,$ with $y_0 \in \mathbb{Z}$.
%\]
\end{lem}%\vskip-1ex
The proof is reported in App.~\ref{sec:proofs}.
%The proof is reported in \cite{techrep} for the sake of space.

%%%%%%%%%%%%%%%%%%%%%%%%%%%%%%%%%%%%%%%%%%%%%%%%%%%%%%%%%%%%%%%%%%%%%%%%%%%%%%%%%%%%%%%%%%%%%%%%%%%%%%%%%%%%%%%%%%%%%%%%
\begin{algorithm}[t]
\caption{{\tt PARALLEL WAKE-UP}}\label{alg:algorit2}
%\small
\begin{algorithmic}[1]
\REQUIRE {$L_i \leq U_i$}
\FORALL {$i \in V$}
\STATE $n_i=\mbox{{\tt PERIOD}}(L_i,U_i)$ $\rightarrow Default:$ $A_i=\{\alpha_i\}$ 
%\STATE {\bf default} $A_i=\{\alpha_i\}$ \{Default phase set\}
%\STATE $NG_i \leftarrow \Ng{i} $ \COMMENT{Enlist the neighboring nodes of $i$}
\STATE $f(i),t(i) \leftarrow 0$ \COMMENT{Fix the initial value to flags $f$ and $t$}
\ENDFOR
\ENSURE {Node $i$ starts a random timer }
%\STATE Timer expired at node $i$: send a message to stop neighbors' timers in $NG_i$
\STATE Timer expired at node $i$: send a message to stop neighbors' timers in $\Ng{i}$
\STATE {\bf Initialize:}
\STATE $n_i\leftarrow [n_{i},(n_{i_1},\ldots,n_{i_{d_i}})]$ \COMMENT{Use the largest possible period at $i$}
\IF{$f(i)=0$ \mbox{ \bf and} there exists $j \in \Ng_i$ s.t. $f(j) \neq 0$} %\COMMENT{If $i$ was not explored before and a neighbor was explored before}
      \STATE $\alpha_i\leftarrow \alpha_j$ \COMMENT{Assign the neighbour's phase to node $i$}
      \STATE $t(j)\leftarrow 1$  \COMMENT{Node $j$ cannot change its phase anymore}
%%%%%%%%%%%%%%%%%%%%%Ajust the phase%%%%%%%%%%%%%%%%%%%%%%%%%%%%%%%%%%%%%%%%%%%
\ELSIF{ there exist $j,k \in \Ng_i: f(j)=f(k)=2$ \mbox{ \bf and }
  $A_j\cap A_k = \emptyset$ \mbox{ \bf and } conditions of
  Lemma~\ref{prop:adjphase} are satisfied}
  \STATE \mbox{ \bf case }$t(i)=0$ \mbox{ \bf then }$\alpha_i\leftarrow a_{jk} $  %\COMMENT{Assign ajusted the neighbors's phase to node $i$}
  \STATE \mbox{ \bf case }$t(i)=1$ \mbox{ \bf then }$A_{i}=A_{i} \cup \{a_{jk}\}$ %\COMMENT{Add the ajusted neighbors's phase to node $i$}
  \STATE    $A_i \leftarrow A_i \setminus \{a_{ij},a_{ik}\}$ \COMMENT{Delete $a_{ij}$ and $a_{ik}$ from $A_i$}
%  \STATE    $NG_i = NG_i \setminus \{j,k\}$ \COMMENT{Delist $j$ and $k$}       
%%%%%%%%%%%%%%%%%%%%%%%%%%%%%%%%%%%%%%%%%%%%%%%%%%%%%%%%%%%%%%%%%%%%%%%
\ENDIF
\STATE {\bf Visit all neighbors:}
%\WHILE{$NG_i \not = \emptyset$} %\COMMENT{Visit all neighbors}}  
%\STATE $j=NG_i[1]$ \COMMENT{Pick the first neighbor}     
\FORALL{$j \in \Ng{i}$  { \bf and}  $f(j)\not=2$} 
\IF{$f(j)=0$ \COMMENT{If it was not explored before}}
  \STATE $\alpha_j\leftarrow \alpha_i$  \COMMENT{Assign the phase to node $j$}     
  \STATE $f(j)=1$
\ELSIF{$A_{i}\cap A_{j} = \emptyset$}  
  \STATE $A_{j}\leftarrow A_{j} \cup \{\alpha_i\}$      \COMMENT{Add the phase to node $j$} 
\ENDIF  
\STATE $x_{ij}^0=\alpha_i+u \frac{\alpha_j-\alpha_i}{(n_i,n_j)}$
\COMMENT{$u$ computed using extended Euclidean algorithm}
\STATE $x_{ij}=x_{ji}=x_{ij}^0+k [n_i,n_j]$, $k=0,1,\ldots$ \COMMENT{Rendezvous times}     
%\STATE    $NG_i = NG_i \setminus \{j\}$ \COMMENT{Delist $j$}
\STATE $f(j)=2$
\ENDFOR
%\ENDWHILE   
\end{algorithmic}
\end{algorithm}
%%%%%%%%%%%%%%%%%%%%%%%%%%%%%%%%%%%%%%%%%%%%%%%%%%%%%%%%%%%%%%%%%%%%%%%%%%%%%%%%%%%%%%%%%%%%%%%%%%%%%%%%%%%%%%%%%%%%%%%%

The algorithm pseudo-code is reported in Alg.~\ref{alg:algorit2}. Each node $i$ maintains information 
about nodes in its neighborhood, i.e., the node ID, the phases set $A_{j}$, the wake-up period $n_j$ 
and two special flags $f$ and $t$. Special flags are used in order to determine $A_i$ and $n_i$. In particular 
(i) $f(j)=0$ means $j$ unexplored (ii) $f(j)=1$ means node $j$ discovered but $n_j$ has not been fixed yet 
and all neighbors of $j$ have been explored (iii) $f=2$ when node $j$ has been explored and therefore 
$\alpha_j$ and $n_j$ have been fixed. Also, $t(j)=0$ means $j$ that can adjust $\alpha_j$ (ii) $t(j)=1$ means 
$i$ cannot change $\alpha_j$ value because some of its neighbors used $\alpha_j$. 

At the start of the algorithm each node $i$ chooses $n_i$ using {\tt PERIOD} and maintains a list of neighbors 
$\Ng{i}$. Then node $i$ is marked {\em unexplored} and {\em available} to change $\alpha_i$ value ($f[i]=t[i]=0$). 
Hence, each node starts a random timer: when the timer expires node $i$ freezes neighbors' timers 
and determines rendezvous times with all nodes in its neighborhood.% using the Chinese Reminder theorem. 

%Then, visits all neighbors. 
If $i$ is unexplored, i.e., $f(i)=0$, node $i$ checks whether it is in range of an already existing 
node $j$ explored. If so, $i$ just aligns $\alpha_i=\alpha_j$; $j$ is then is marked {\em not available} 
to change its time axis anymore, i.e., $t(j)=1$. Otherwise, 
if $i$ is {\em discovered}, i.e., $f(j)=1$, $i$ checks whether it is in range of two already existing 
$j$ and $k$ already explored without a phase in common between them and if the phase can be adjusted then 
$i$ has two options: (i) if $i$ is {\em available} to change its phase ($t(i)=0$), then $i$ just aligns 
its phase to the novel phase, i.e., $\alpha_i=\alpha_j$, (ii) if $i$ is {\em not available} to 
change its phase ($t(i)=1$), then the novel phase is added to $A_i$.

Then, $i$ visits its neighbors for which no adaptation was performed. If $j\in \Ng{i}$ is {\em unexplored}, $\alpha_j=\alpha_i$ and $j$ is marked {\em discovered}, 
$f(j)=1$. If $j$ is {\em explored} and $A_i\cap A_j=\emptyset$, the phase of $i$ is added to $A_j$. Finally, $i$ marks 
itself {\em explored}.%, $f(i)=2$.

The proof of correctness and the fact that Alg.~\ref{alg:algorit2} has message exchange 
is $O(M)$ is similar to the one for Alg.~\ref{alg:algorit1}:
\begin{thm}
{\tt PARALLEL WAKE-UP}:
\begin{enumerate}
\item[  i.] produces a wake-up schedule; 
\item[iii.] it requires $O(M)$ messages for a distributed implementation.
\end{enumerate}
\end{thm}

In principle, since multiple phases may be added at one node, it is possible that 
the output of the algorithm is a schedule that is not feasible. A simple bound on 
the average duty cycle (DC) is provided by the following
\begin{thm}
Under Alg.~\ref{alg:algorit2} $DC<\frac {d-1}{L}$, where $L=min L_i$ and $d$ is the average node degree
\end{thm}
%\begin{IEEEproof}
\proof
The algorithm adds at most $d_i-1$ phases per node: $\frac 1N \sum_{i=1}^N \frac{|A_i|}{n_i}\leq \frac 1{NL}\sum_{i=1}^N (d_i-1)=\frac {d-1}{L}$
\endproof
%\end{IEEEproof}

\begin{rem}
Observe that the above algorithm can be used to adapt the wake-up scheduling dynamically in case one or more nodes 
join the network. In particular, it is possible to combine the two algorithms in order to initialize a 
wake-up schedule network-wide with {\tt BFS WAKE-UP}, and use {\tt PARALLEL WAKE-UP} to add new nodes to existing schedules
 or to repair locally a schedule, should a node fail or loose synchronization.
\end{rem}

%%%%%%%%%%%%%%%%%%%%%%%%%%%%%%%%%%%%%%%
\subsection{Example}
%%%%%%%%%%%%%%%%%%%%%%%%%%%%%%%%%%%%%%%

We reported in Fig.~\ref{fig:example} a simple example 
illustrating the {\tt BFS WAKE-UP} and of the {\tt Parallel WAKE-UP} 
algorithm on a sample graph with $7$ 
nodes. In that instance $U=20$, $L_i=(2,3,9,7,11,5,2)$, 
$\alpha_i=(1,3,7,9,6,1,0)$ and $B=\{2\}$. Several 
wake-up periods coexist: $4,8,16$.

\begin{figure*}[t]
\centering
%\hskip1mm
\begin{minipage}{4cm}
\includegraphics[width=2cm]{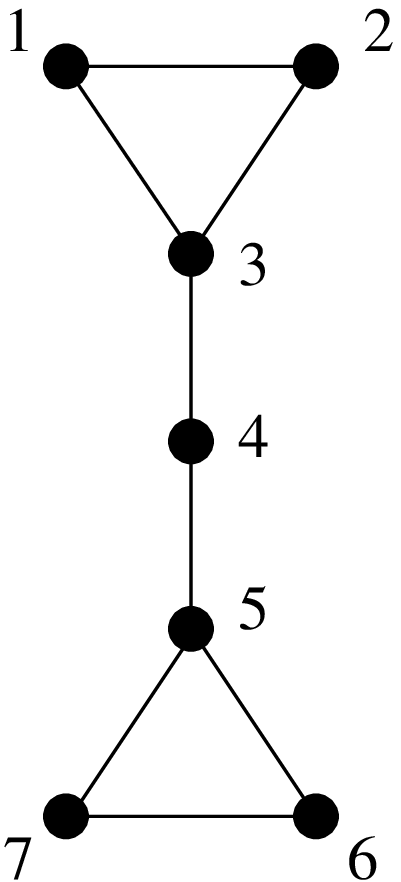}
\end{minipage}
\begin{minipage}{4cm}{%
%\begin{scriptsize}
\begin{tabular}{|l|c|c|c|c|c|c|}
\hline
\textbf{$i$ } & \textbf{$\alpha_i$} & \textbf{$L_i$}& \textbf{$U_i$}& \textbf{$p_i^{q_i}$}&\textbf{$(\alpha_i,n_i)$}& \textbf{$(\alpha_i,n_i)$}\\
&  &   & &  &\textbf{BFS} & \textbf{PARALLEL}\\
\hline
\hline
$1$ & $1$ &$2$ &$20$ &$2$  & $(1,4)$ & $(3,4)$\\
$2$ & $3$ &$3$ &$20$ &$4$  & $(1,4)$ & $(3,4)$\\
$3$ & $7$  &$9$ &$20$ &$16$ & $(1,16)$ & $(3,16)$\\
$4$ & $9$ &$7$ &$20$ & $8$  & $(1,4)$ & $(\{1,3\},4)$\\
$5$ & $6$  &$11$ &$20$ & $16$ & $(1,16)$ & $(1,16)$\\
$6$ & $1$  &$5$ &$20$ &$8$  & $(1,8)$ & $(1,8)$\\
$7$ & $0$  &$2$ &$20$ &$2$  & $(1,8)$ & $(1,8)$\\
\hline
\end{tabular}
%\end{scriptsize}
}\end{minipage}
\caption{Example. {\tt BFS WAKE-UP} and {\tt PARALLEL WAKE-UP} on a sample topology. The table reports on the final output of the algorithms.}\label{fig:example}
\end{figure*}

The {\tt BFS WAKE-UP} run generates the list of visited 
nodes $(1,2,3,4,5,6,7)$; observe that every node aligned to 
$\alpha_1=1$. In the {\tt Parallel WAKE-UP} node $2$ and $7$ perform 
 the algorithm operations first, then $1$, $5$, $6$ and finally $4$. 
As we can see from the example, {\tt Parallel WAKE-UP} assigns 
heterogeneous phases; at node $4$ phase adaptation according 
to Thm.~\ref{prop:adjphase} is not possible so that finally 
$A_4=\{1,3\}$.

%%%%%%%%%%%%%%%%%%%%%%%%%%%%%%%%%%%%%%%%%%%%%%%%%%%%%%%%%%%%%%%%%%%%%%%%%%%%%%%%%%%%%%%%%%%%%%%%%%%%%%%%%%%%%%%%%%%%%%%%
\subsection{Particular cases}\label{sec:particular}
%%%%%%%%%%%%%%%%%%%%%%%%%%%%%%%%%%%%%%%%%%%%%%%%%%%%%%%%%%%%%%%%%%%%%%%%%%%%%%%%%%%%%%%%%%%%%%%%%%%%%%%%%%%%%%%%%%%%%%%%

There exists specific cases when strictly tight wake-up schedules are attained easily.
\begin{lem}
Let for any node $i$ be $B|L_i$, $B|U=U_i$ and $L_i\leq U$, then there exists a strictly tight wake-up schedule; it can be attained
using {\tt BFS WAKE-UP}. 
\end{lem} 
\begin{proof}
Consider any link $ij$ and let $n_i=L_i$ and $n_j=L_j$. The statement follows immediately 
from the Chinese Reminder theorem since $[n_i,n_j]|U$ and one can use 
{\tt BFS WAKE UP} to assign the same phase to all nodes in the network. 
\end{proof}

In particular the following example shows a simple application of the above
case. Consider a tree topology $T$ and assume nodes are labeled in order to respect the partial order induced 
by the distance from the root node $1$. Let $L_0$ be associated to the root node, and $L_k=L_0p^k$ be associated 
to nodes at level $k$. A {\em feasible} schedule can be constructed as follows: root node $0$ wakes at times $x_0=0 \mod L_0$.
Node $j$ wakes up at $x_j=\alpha_0 \mod p^{g(r_j)} L_0 $ where $r_j$ is the level occupied by node $j$ and $g$ is any positive integer 
function. In this case $U=L_0p^{\max g(r_j)}$.

% In particular, assume that $L_i\leq U=\min \{ U_j,\, j=1,\ldots, N\}$, for $i=1,\ldots,N$. It is easy to see 
% that in this case, the schedule $(A_i,n_i)=(0,U)$, for $i=1,\ldots,N$ is a trivial 
% solution of the WAKE-UP problem, and it can be determined in polynomial time with the following procedure: 
% \begin{itemize}
% \item Determine $U=\min \{ U_j,\, j=1,\ldots, N\}$;
% \item Align phases such in a way that $\alpha_i=\alpha_j$ for $i,j=1,\ldots, N$ 
% \end{itemize}
% However, this approach has several limitations. The first step requires an exchange on the order of $\Theta(MN)$ messages, 
% i
%.e. it is quadratic in the order of exchanged messages in the network. Also, we aim at reducing the total delay drift. 

% The algorithm is divided into two parts. First, assume that $U$ factorizes into a set of ``small'' 
% factors: if it does not, consider a smaller $U\geq L_i$, $i=1,\ldots,N$ that does so. Say $U=\prod_{i=1}^{O_R}p_i^{q_i}$. 

%%%%%%%%%%%%%%%%%%%%%%%%%%%%%%%%%%%%%%%
%%%%%%%%%%%%%%%%%%%%%%%%%%%%%%%%%%%%%%%%%%%%%%%%%%%%%%%%%%%%%%%%%%%%%%%%%%%%%%%%%%%%%%%%%%%%%%%%%%%%%%%%%%%%%%%%%%%%%%%%

%%%%%%%%%%%%%%%%%%%%%%%%%%%%%%%%%%%%%%%%%%%%%%%%%%%%%%%%%%%%%%%%%%%%%%%%%%%%%%%%%%%%%%%%%%%%%%%%%%%%%%%%%%%%%%%%%%%%%%%%
%%%%%%%%%%%%%%%%%%%%%%%%%%%%%%%%%%%%%%%%%%%%%%%%%%

\section{Performance Evaluation}\label{sec:numerical}
\begin{figure}[t]
  \centering
  \subfloat[]{\includegraphics[width=0.45\textwidth]{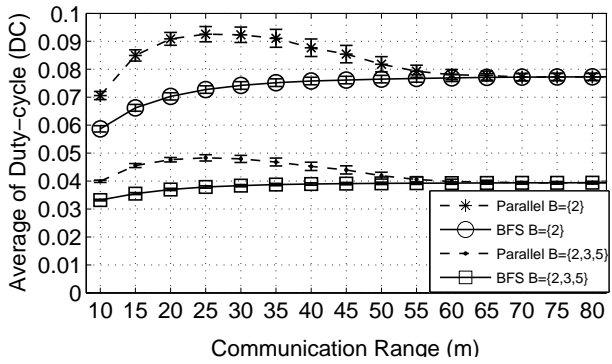}}\quad              
  \subfloat[]{\includegraphics[width=0.45\textwidth]{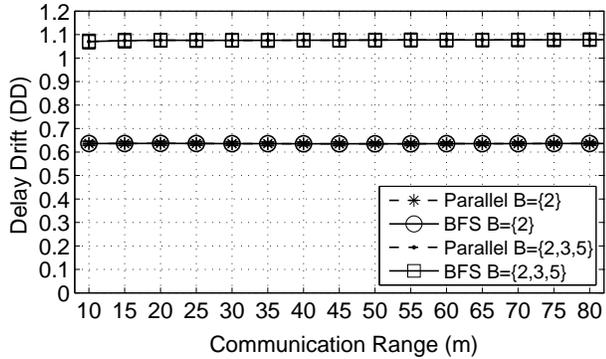}}\\
  \subfloat[]{\includegraphics[width=0.45\textwidth]{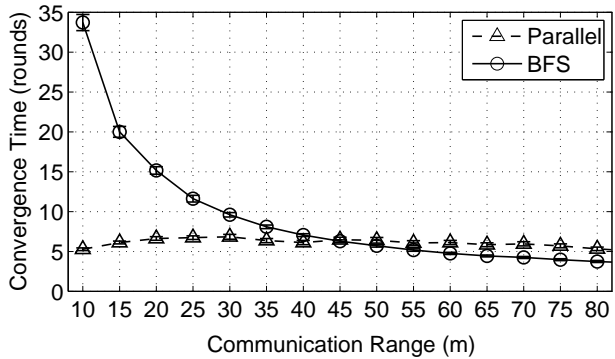}}\quad
\caption{Performance of {\tt BFS WAKE-UP} and {\tt PARALLEL WAKE-UP}: results for $L_i\in [1,35]$ and $U_i\in [50,100]$: 
a) average duty cycle (DC) b) average delay drift (DD) and c) convergence time.}
\label{fig:set_1}
\end{figure}
In this section, we report on the outcomes of experiments for the {\tt BFS WAKE-UP}
and {\tt PARALLEL WAKE-UP} algorithms. In the first set of experiments,
performed using  Matlab$^{\scriptsize \textregistered}$, we tested the algorithms on random topologies.

We considered a square playground of side $100$ m with $200$ nodes
randomly deployed according to a uniform distribution. Nodes were
assumed to be connected if their mutual distance was below a given
value (communication range), varied in the range $10-80$m. For each
setting, $100$ runs were performed. For the construction of the
schedules, two factor bases were considered, $B=\{2\}$ and $B=\{2,3,5\}$.
The constraints were set by taking, for each node, $L_i$ uniformly
distributed in $[1,35]$ and $U_i$ uniformly distributed in $[50,100]$.

We considered as performance metric the average duty cycle, as
introduced in the previous section, and the average normalised delay drift,
defined as:
\begin{eqnarray}
DD=\frac 1{2M} \sum\limits_{i=1}^N \sum\limits_{j \in N_g(i)}\frac{[n_i,n_j]}{U_i}
\end{eqnarray}
It is worth remarking that $DD>1$ implies that some constraints on the
delay were violated. 

Results are reported in Fig.~\ref{fig:set_1}. In terms of duty cycle,
{\tt BFS WAKE-UP} outperforms {\tt PARALLEL WAKE-UP}, as expected, since the second one may add
multiple phases per node. The performance difference reduces in dense
networks (i.e., with large transmission range).  A closer look at the
results showed that, for the parameters chosen, the schedules generated
by {\tt PARALLEL WAKE-UP} were not feasible in less than $1\%$ of the cases.
In terms of average normalised delay drift,
the two algorithms attain the same performance. {\tt BFS WAKE-UP} is much slower in
converging in sparse networks, but its performance increase (and,
actually, outperforms {\tt PARALLEL WAKE-UP }) in very dense networks. As it
may be seen, the use of a richer basis improves
significantly the performance in terms of duty cycle, but it also leads to
worse delay drift performance.

In general, we may conclude that {\tt PARALLEL WAKE-UP} represents a meaningful
choice in most situations, presenting performance close to {\tt BFS} but
with a shorter convergence time.
\begin{table}[t]
\centering
\begin{tabular}{|l|c|c|c|c|}
\hline
\textbf{ Violation    }     & $B$ & $[75,100]$ & $[60,100]$ & $[45,100]$ \\
\hline
\hline
\textbf{Magnitude} &$\{2\}$ & $0$  & $2.04 \pm 0.07$ & $9.32 \pm0.26$\\
\hline
\textbf{Magnitude} & $\{2,3,5\}$  & $0$  & $0$ & $88.38 \pm1.84$\\
\hline
\textbf{Percentage} & $\{2\}$ & $0$  & $1.25\% \pm0.12$& $5.43\% \pm0.27$\\
\hline
\textbf{Percentage} & $\{2,3,5\}$ & $0$  & $0$ & $9.36\% \pm0.43$\\
\hline
\end{tabular}\caption{Performance of {\tt BFS WAKE-UP} and {\tt PARALLEL WAKE-UP}, fraction and magnitude of violations. Settings as in Fig.~\ref{fig:set_1}.}\label{tab:viol}
\end{table}

{\noindent \em Discrete events simulation:} In a second set of experiments, 
we implemented the proposed techniques
in an event-based simulator, Omnet++, to assess the advantage of 
supporting heterogeneous wake-up schedules. We considered
$50$ sensors deployed over a $100m \times 100m$ area, with
communication range of $10m$ and set the time slot duration to $100$ms. We
assumed that packets were generated at each node every $30$ s, and
that all packets were routed towards a sink, located in $(1,1)$. A
tree was constructed at the beginning of the simulation to route
packets to the sink.

As sensors far from the sink relay less traffic, their bounds on
energy consumption can be tightened. We considered a situation in
which the constraints $L_i$ depended on the distance from the sink as follows: 
$L_i=2$ for nodes at distance $\leq 2$ hops, $L_i=4$ for nodes at
distance of $3$ and $4$ hops and so on. As performance metrics we
considered the average packet delay (from source to sink) and the
lifetime of the nodes. The latency was computed up to the time at which the first node
died. We compared the performance attained by {\tt BFS}
({\tt PARALLEL} turned out to offer the same performance) with those
obtained in the absence of sleep mode ('No Power Saving') and those
obtained when all nodes go to sleep synchronously for one slot every
two ('Uniform'). The results are reported, plotted against the distance
from the sink, in Fig.~\ref{fig:set_4}. As it can be seen, {\tt BFS} is able
to effectively prolong the lifetime of sensors far from the sink,
while at the same time limiting the increase in packet latency.
\begin{figure}[t]
  \centering
  \subfloat[]{\includegraphics[width=0.45\textwidth]{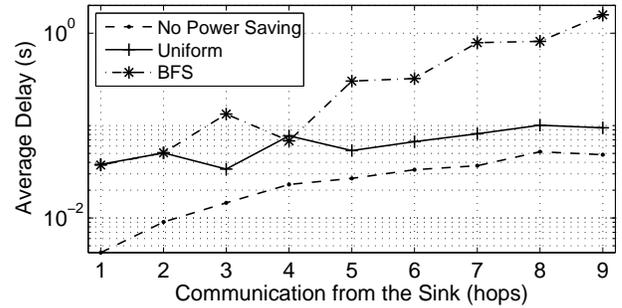}}\quad               
  \subfloat[]{\includegraphics[width=0.45\textwidth]{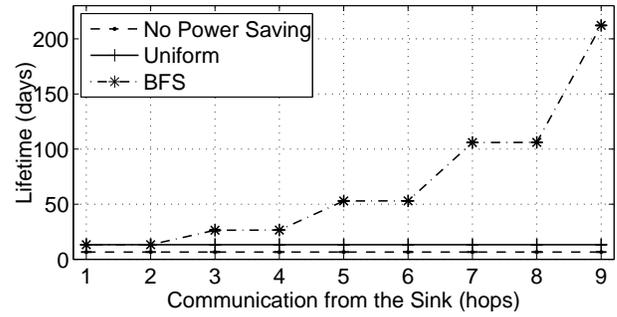}}\quad 
\caption{Discrete event simulation: a) average node lifetime as a function of the distance from the sink, b) 
average delay.}\label{fig:set_4}
\end{figure}

%%%%%%%%%%%%%%%%%%%%%%%%%%%%%%%%%%%%%%%%%%%%%%

%%%%%%%%%%%%%%%%%%%%%%%%%%%%%%%%%%%%%%%%%%%%%%%%%%
%%%%%%%%%%%%%%%%%%%%%%%%%%%%%%%%%%%%%%%%%%%%%%%%%%%%%%%%%%%%%%%%%%%%%%%%%%%%%%%%%%%%%%%%%%%%%%%%%%%%%%%%%%%%%%%%%%%%%%%%

%%%%%%%%%%%%%%%%%%%%%%%%%%%%%%%%%%%%%%%%%%%%%%%%%%%%%%%%%%%%%%%%%%%%%%%%%%%%%%%%%%%%%%%%%%%%%%%%%%%%%%%%%%%%%%%%%%%%%%%%
%%%%%%%%%%%%%%%%%%%%%%%%%%%%%%%%%%%%%%%%%%%%%%%%%%

\section{Conclusions}\label{sec:conclusions}
In this paper we introduced a framework for the synchronization of wake-up schedules, which plays 
a central role for energy saving in wireless networks, under the joint request to satisfy per node 
energy and delay constraints. We showed that the related {\tt WAKE-UP} problem has a strong relation 
with integer factorization, which adds a novel algebraic perspective.
 
This work left out several interesting directions. In particular, for relatively small numbers, factorization 
can be done in practice \cite{giblin}. In turn, the optimal choice of the factor basis and of periods 
is a key problem in order to limit the drift of the rendezvous times. 

Furthermore, mapping end to end delay onto bounds on the wake-up period has been only touched briefly in 
App.~\ref{app:boundsdym}. However, while there exist several works in literature tackling the general problem 
of bounding the end-to-end delay \cite{Chlamtac,LeBoudec}, a network calculus accounting 
for power saving duty cycles has not been proposed so far. Relating the {\tt WAKE-UP} problem to those 
works is part of future work.

%%%%%%%%%%%%%%%%%%%%%%%%%%%%%%%%%%%%%%%%%%%%%%%%%%
%%%%%%%%%%%%%%%%%%%%%%%%%%%%%%%%%%%%%%%%%%%%%%%%%%%%%%%%%%%%%%%%%%%%%%%%%%%%%%%%%%%%%%%%%%%%%%%%%%%%%%%%%%%%%%%%%%%%%%%%

%%%%%%%%%%%%%%%%%%%%%%%%%%%%%%%%%%%%%%%%%%%%%%%%%%%%%%%%%%%%%%%%%%%%%%%%%%%%%%%%%%%%%%%%%%%%%%%%%%%%%%%%%%%%%%%%%%%%%%%%

\bibliographystyle{IEEEtran}
\bibliography{bibliography}

% %%%%%%%%%%%%%%%%%%%%%%%%%%%%%%%%%%%%%%%%%%%%%%%%%%%%%%%%%%%%%%%%%%%%%%%%%%%%%%%%%%%%%%%%%%%%%%%%%%%%%%%%%%%%%%%%%%%%%%%%
% \appendices
% %%%%%%%%%%%%%%%%%%%%%%%%%%%%%%%%%%%%%%%%%%%%%%%%%%%%%%%%%%%%%%%%%%%%%%%%%%%%%%%%%%%%%%%%%%%%%%%%%%%%%%%%%%%%%%%%%%%%%%%%

\include{appendices}
\appendix

\section{Proofs}\label{sec:proofs}
\subsection{Proof of  Thm.~\ref{thm:BFS}}

%\begin{IEEEproof}
\proof
i. Consider node $i$ and assume that the algorithm started at node $1$ for simplicity: at the end of the algorithm 
it holds $(A_i,n_i)=(\{a_1\},[(n_{i},n_{i_1}),\ldots,(n_{i},n_{i_{N_i}})\})$. Notice that $a_i-a_{j_i}=0$ for 
all $j_i\in N_i$ and also 
\begin{eqnarray}
   ([n_{i},n_{i_1}],\ldots,[n_{i},n_{i_{N_i}}])=[n_{i},(n_{i_1},\ldots,n_{i_{N_i}})]
\end{eqnarray}
Hence, every pair $(a_i,n_i)$ solves the modular congruence system according to the 
Chinese Remainder theorem. 

ii. For each node $j$, $\Ng{j}$ operations are performed, including the calculation of 
$[n_,n_j]$ , which requires on $O(\log^3 \max(n_i,n_j))$ operations. The overall complexity of the algorithm 
is then bounded as $\sum_{i=1}^{N} \sum_{j\in \Ng{i} } \log^3 \max(n_i,n_j) =2E \log^3 U$

iii. The statement follows immediately by a double counting argument since every node $i$ needs to 
exchange at most $2$ messages with each neighbor in order to visit neighbors, communicate its phase 
and receive acknowledgment.
%\end{IEEEproof}
\endproof

\subsection{Proof of  Lemma ~\ref{prop:adjphase}}

%\begin{IEEEproof}
\proof
A novel phase $a$ can be given to node $j$ iff the following system solves
\begin{eqnarray}\label{eq:system}
x&&= a_{ij} \mod n_i = a \mod n_j \nonumber \\
 &&= a_{kj} \mod n_k 
\end{eqnarray}
from which $(n_k,n_j)|a_{kj}-a$ and $(n_i,n_j)|a-a_{ij}$. Thus, if a novel phase can be assigned
then a solution exists for
\[
            c=a_{kj} - a_{ij}=x \, (n_j,n_k) +y \, (n_i,n_j),\quad x,y \in \Z
\]
From (\ref{eq:system}), $d=(n_i,(n_j,n_k))|a_{kj} - a_{ij}$: let $d e =a_{kj} - a_{ij}$ consider a solution $x_0,y_0$
\[
d=(n_j,n_k) x_0 + (n_i,n_j) y_0 ,\quad x_0,y_0 \in \Z
\]
then 
\[
c=de= (n_j,n_k)e x_0 + (n_i,n_j)e y_0 
\]
the novel phase writes 
\[
a=a_{kj}-(n_j,n_k) e x_0
\]
Observe that $a-a_{ij}=a_{kj}-(n_j,n_k) e x_0-a_{ij}=(n_i,n_j)e y_0$ which concludes the proof.
%\end{IEEEproof}
\endproof

\subsection{Continuous--Time Analogue}\label{sec:continuous}
Here we consider wake-up scheduling in continuous time, under the 
additional requirement that the difference of phase at different nodes is rational (observe that this is a practical 
requirement since numerical representation is rational). We conclude that the continuous 
time approach plus the request for a rational difference of phase at two nodes is equivalent to 
operate wake-up scheduling in the discrete time domain. In fact, consider the 
general case of a network where wake-up scheduling is operated and nodes are assigned 
phases $\alpha_i$ and period $T_i$. It holds for any node $i$:
\[
K_{ij} T_i+\alpha_{i}=K_{ji} T_j+\alpha_{j}, \; j\in \Ng{i}
\]   
Now, assume $G$ is connected, and consider $N$ paths $P_1,P_2,\ldots,P_N$ from node $1$ to nodes $j=1,\ldots,N$, 
respectively: it is easy to see that the following set of equalities hold
\[
T_j/T_1=q_{j1}+\sum_{rs\in P_j} p_{rs} \alpha_{rs}
\]
where $\alpha_{rs}= \alpha_{r} -\alpha_{s}$, whereas $q_{j1}$ and $p_{rs}$ are rational. Thus, it follows that also $T_j/T_1$ is 
rational. This is equivalent to admit that the overall system time, i.e., the period of each node $T_i$ can be reduced to a 
convenient multiple of a certain $T_0=T_1/N_0$.

\subsection{Standard isomorphism}\label{sec:isomorphism}
A basic result in algebra that provides a natural link of the Chinese Remainder theorem with the problem of factorization 
of integers. With standard notation denote $\ZnZ{n}$ the quotient ring of integers modulus $n$ \cite{langu}. 
Define the standard isomorphism
\begin{eqnarray}
\Psi:& \ZnZ{n} \times \ZnZ{m} \longrightarrow \ZnZ{[m,n]}\nonumber\\
     &   (a,b) \rightarrow c     \nonumber
\end{eqnarray}
where $c$ solves for 
\begin{eqnarray}
c &= a \mod n \nonumber \\
c &= b \mod m \nonumber
\end{eqnarray}
Assume now that $m,n$ are two coprime integers, i.e.,  $(m,n)=1$; it is easy to see that $\Psi$ is an 
isomorphism. Translating with the terminology used in wake-up scheduling used throughout this paper, 
this means that for every wake-up schedule $\{(a,n),(b,m)\}$, there exists a unique rendezvous time $c_0$ such
 that $c=c_0 \mod [m,n]$ for any other rendezvous time $c$ generated by the schedule. Viceversa, given a non-empty 
set of rendezvous times for two nodes having coprime periods, there exists a {\em unique} pair of phases $a$ and $b$ 
that corresponds to such set. This also means that phases are automatically determined by the set of rendezvous times 
and the coprime periods, i.e., the only information needed is that $mn=[m,n]$, i.e., $(m,n)=1$. This provides another 
interpretation of the reason why {\tt WAKE-UP} is as hard as {\tt INTEGER FACTORIZATION}, as proved in Sec.~\ref{sec:complexity}.

\subsection{On the Dimensioning of the Upper Bounds}\label{app:boundsdym}
In general, application--level constraints for the application setting we target are expected to be given in terms of:
\begin{itemize}
\item {\it Expected lifetime of the network.} This can be translated in terms of an upper bound on the duty cycle of single nodes. This is reflected in a constraint of the form like (\ref{eq:constr_lower}).
\item {\it End-to-end delay.} Bounds will be given in terms of the delay incurred between the generation of a message at a node (which corresponds, in our framework, to the reading of a sensor) and the delivery at the appropriate sink, where such data may be processed or transferred, via some other communications technology, to a remote processing center.
\end{itemize}
For the latter case, we need to identify means for translating end-to-end delay constraints into node-level delay constraints of the form (\ref{eq:constr_upper}). The problem is then to translate a number of constraints, given on routes, to constraints on the single links constituting the routes. 

In order to do so, we start by introducing some notation. We define by $R_i$ the route starting from node $i$ and ending in one of the sink nodes. We say that node $j$ lies along the route $R_i$ if all traffic originating from $i$ passes through $j$ in order to reach the sink, and write $j \in R_i$. The length (in hops) of route $R(i)$ is denoted by $|R(i)|$. Conversely, we denote by $\mathcal{F}_i$ the set of routes passing through node i. %The degree of node $i$ is indicated by $deg(i)$. 
The worst-case delay associated to route $R_i$ is denoted by $D(R_i)$. %The term 'worst-case'  accounts for the possible generation time of messages at node $i$. 
The delay constraint on route $R_i$ is expressed as
$D(R_i)\leq \eta_i$, 
where $\eta_i$ is a constant.

The difficulty of the problem lies in the fact that the constraint on the energy expenditure (\ref{eq:constr_lower}) impacts the minimal latency experienced by a message passing through a given node. 

We assume that (i) routes are static throughout the lifetime of the network (ii) nodes arriving during an active period are sent at the next active period occurrence (iii) the duration of a time slot is sufficient for transmitting all messages queued at any given node. 

Now we consider the latency experienced at a given node and imposed by the duty cycle. Under the aforementioned assumptions, the worst-case delay incurred by a message passing through node $i$ can be lower bounded by $\frac{n(i)}{|A_i|}+1$. The term 'worst-case' here refers to two factors. One, to the arrival time of the message, whereby the worst case coincides with a message arriving right after the beginning of an active slot. Second, to the rendez-vous points. In order to be able to move the message one hop further along the route, it is not sufficient for node $i$ to be in the active state, but, given the system model presented in Sec.~\ref{sec:model}, it requires also the next hop node to be in the active state. In general, the delay will be a multiple of $\frac{n(i)}{|A_i|}$ plus the time slot necessary to forward the message. \footnote{It is worth remarking that the bound is exact if $|A_i|=1$. For $|A_i|>1$ it is an approximation of the actual one, derived assuming that active slots are uniformly distributed within the wake-up period.}%The solutions provided by our approach respect the constraint $|A_i|\leq deg(i)$.\footnote{This is discussed in detail in~\ref{puppa}.} Hence we have that the worst-case delay incurred at node $i$ can be lower-bounded by $\frac{n_i}{deg(i)}+1$.
Given condition (\ref{eq:constr_lower}), this implies that such worst-case delay is lower bounded by $L_i+1$. 

Clearly, a set of feasible $U_i$ has to satisfy $U_i\geq L_i+1$ for all nodes $i$. Furthermore, in order to be compatible with the end-to-end delay constraints it has to be $\sum_{j \in R_i} U_j \leq \eta_i$ for all routes $R_i$.

We thus need to find a solution of the following system of inequalities:
\begin{equation}\label{eq:systeme2e}
\left\{
\begin{array}{l}
\sum\limits_{j \in R_1} U_j \leq \eta_1,\\
\vdots\\
\sum\limits_{j \in R_N} U_j \leq \eta_N,\\
U_1 \geq 1+L_1,\\
\vdots\\
U_N \geq 1+L_N.
\end{array}
\right.
\end{equation}

The system has a solution if and only if the following condition is satisfied:
\begin{equation}\label{eq:e2econd}
|R_i| + \sum_{j \in R_i}L_j \leq \eta_i\,\, \forall i=1,\ldots,N.
\end{equation}

If such a condition is respected, a feasible solution can be found as follows. We associate to route $R_i$ a ``delay budget'' $\eta_i$. We consider how much is consumed by the duty cycle at the various nodes along the route $R_i$, which, using the bound outlined above, will be taken as $|R_i| + \sum_{j \in R_i}L_j $. We then distribute the remaining ``delay budget'' in a uniform manner across all nodes of the route $R_i$. We do this for all routes. Then, for each node $i$, we consider as $U_i$ the minimum value of such quantity among all routes passing through it. In symbols:
\begin{equation}\label{eq:uniform}
U_i=L_i+1+\min\limits_{R_j \in \mathcal{F}_i}\left[ \frac{\eta_j-|R_j|-\sum\limits_{k\in R_j}L_k}{|R_j|}\right].
\end{equation}

It is trivial to prove that under (\ref{eq:e2econd}) solutions of the form (\ref{eq:uniform}) satisfy $U_i\geq L_i+1 \,\, \forall i$. In order to prove that the other conditions in (\ref{eq:systeme2e}) are satisfied we note that:
\begin{align}
\sum\limits_{i\in R_j} U_i=\sum\limits_{i \in R_j} \left\{ L_i + 1 + \min\limits_{R_k \in \mathcal{F}_i}\left[ \frac{\eta_k-|R_k|-\sum\limits_{h \in R_k}L_h}{|R_k|}\right]\right\}\leq\\
\leq \sum\limits_{i \in R_j} L_i + |R_j| + \sum\limits_{i \in R_j}  \frac{\eta_j-|R_j|-\sum\limits_{h \in R_j}L_h}{|R_j|}=\\
=\sum\limits_{i \in R_j} L_i + |R_j| + \eta_j - |R_j| - \sum\limits_{h \in R_j} L_h =\eta_j,
\end{align}
where we used the fact that if $i \in R_j$ then $R_j \in \mathcal{F}_i$ and that $\min\limits_{i \in S} \{x_i\} \leq x_j \,\, \forall j \in S$.

We recall that the method outlined above provides only {\it one} of the possible solutions of the system (\ref{eq:systeme2e}). In particular, there is no reason for which a uniform distribution of the remaining delay budget along a route is optimal, though this seems reasonable in a variety of settings. 

For the special case in which $L_i=L \, \forall i$ and $\eta_i=\eta \, \forall i$, i.e., all nodes have the same energy constraint and all messages should be delivered to a sink within the same deadline, regardless of their distance, (~\ref{eq:uniform}) simplifies as:
\begin{align}
U_i=L +1+\min\limits_{R_j \in \mathcal{F}_i}\left[ \frac{\eta-|R_j|-\sum_{k\in R_j}L
}{|R_j|}\right]=\\
=L+1+\min\limits_{R_j \in \mathcal{F}_i}\left[ \frac{\eta-|R_j|\cdot (L+1)
}{|R_j|}\right]=\\
=\frac{\eta}{\max\limits_{R_j \in \mathcal{F}_i} |R_j|}
\end{align}

In the uniform case, therefore, the bounds $U_i$ depend only on the maximal length of the routes passing through node $i$.

%%%%%%%%%%%%%%%%%%%%%%%%%%%%%%%%%%%%%%%%%%%%%%%%%%%%%%%%%%%%%%%%%%%%%%
%%%%%%%%%%%%%%%%%%%%%%%%%%%%%%%%%%%%%%%%%%%%%%%%%%%%%%%%%%%%%%%%%%%%%%%%%%%%%%%%%%%%%%%%%%%%%%%%%%%%%%%%%%%%%%%%%%%%%%%%

\end{document}